\documentclass[aps,pra,twocolumn,superscriptaddress,nofootinbib]{revtex4-2}

\usepackage{amsmath,amssymb,mathtools,bm,amsthm}
\usepackage{microtype}
\usepackage{hyperref}
\hypersetup{hidelinks}

\newtheorem{theorem}{Theorem}
\newtheorem{lemma}{Lemma}
\newcommand{\HA}{\mathcal H_A}
\newcommand{\HB}{\mathcal H_B}
\newcommand{\SA}{\mathcal S_A}
\newcommand{\SB}{\mathcal S_B}
\newcommand{\E}{\mathcal E}
\newcommand{\B}{\mathcal B}
\newcommand{\id}{\mathbb I}
\newcommand{\ketbra}[1]{|#1\rangle\!\langle #1|}

\begin{document}

\title{Localizable Bipartite Rank-One Ideal Measurements Have a Block-Replicated Nice-Bell Structure}

\author{Ahmed Younis}
\affiliation{Independent Researcher}

\date{August 27, 2026}

\begin{abstract}
We prove a complete structural characterization of finite-dimensional bipartite rank-one ideal projective measurements that are localizable without communication in the sense of Akibue and Miyazaki. Up to local-unitary equivalence, every such measurement basis is obtained by replicating a single nice Bell basis across equal-dimensional local subspace blocks. The necessity proof combines the causal block structure of complete measurements established by Beckman, Gottesman, Nielsen, and Preskill with their eigenstate-composition theorem for localizable superoperators. A reference block is first forced to be a nice Bell basis; the same theorem then propagates that basis consistently along every row, column, and interior block. The converse follows from the explicit localization protocol of Akibue and Miyazaki. This establishes their Conjecture 1 and gives a protocol-independent classification of bipartite rank-one ideal measurements in this setting.
\end{abstract}

\maketitle

\section{Introduction}

A joint quantum measurement may respect relativistic no-signaling while nevertheless fail to admit an implementation by spatially separated parties who share entanglement but do not communicate. Beckman, Gottesman, Nielsen, and Preskill (BGNP) formalized this distinction in terms of \emph{causal} and \emph{localizable} quantum operations and derived strong structural restrictions on complete projective measurements \cite{Beckman2001}. More recently, the localization problem has been revisited from the perspective of entanglement cost and protocol-independent algebraic constraints \cite{Pauwels2025,Akibue2026}.

Akibue and Miyazaki introduced a localization criterion for instruments and studied ideal projective measurements, where the projected post-measurement state must be reproduced in addition to the classical outcome \cite{Akibue2026}. Their Definition 4 is stronger than the BGNP notion for the associated decoherence map: localizability of an ideal instrument implies BGNP-localizability of the sum of its instrument elements. They also gave an explicit localization protocol for bases of the block-replicated form
\begin{equation}
\left\{W_p^A\otimes W_q^B |M_i\rangle\!\rangle\right\}_{i,p,q},
\label{eq:blockform}
\end{equation}
where the local Hilbert spaces decompose into equal-dimensional subspaces, $\{|M_i\rangle\!\rangle\}_{i=1}^{d^2}$ is a nice Bell basis on one reference block, and $W_p^A,W_q^B$ are unitary isomorphisms between the reference subspaces and the other blocks. They conjectured that, up to local-unitary (LU) equivalence, Eq.~\eqref{eq:blockform} is not only sufficient but necessary for every bipartite rank-one ideal measurement to be localizable \cite{Akibue2026}.

Here we prove that conjecture. The argument is short because the BGNP eigenstate-composition constraint, when combined with their causal block decomposition, removes the freedom to twist different blocks independently.

\section{Setup and two BGNP ingredients}

Let $\{P_x\}_x$ be a rank-one projective measurement on $\HA\otimes\HB$, with $P_x=\ketbra{\phi_x}$, and let
\begin{equation}
\E(\rho)=\sum_x P_x\rho P_x
\label{eq:dephase}
\end{equation}
be its decoherence superoperator. Suppose the associated ideal instrument is localizable in the sense of Definition 4 of Ref.~\cite{Akibue2026}. Then $\E$ is localizable in the BGNP sense and hence causal \cite{Akibue2026,Beckman2001}.

For a causal complete rank-one measurement, BGNP show that there is a common integer $d$ and decompositions
\begin{equation}
\HA=\bigoplus_{p=1}^{r_A}\SA^p,\qquad
\HB=\bigoplus_{q=1}^{r_B}\SB^q,
\label{eq:blocks}
\end{equation}
with $\dim\SA^p=\dim\SB^q=d$, such that the measurement basis splits into blocks
\begin{equation}
\B_{pq}\subset \SA^p\otimes\SB^q,
\qquad |\B_{pq}|=d^2,
\label{eq:blockbasis}
\end{equation}
and every vector in every block is maximally entangled across the corresponding $d\times d$ subspace pair \cite{Beckman2001}.

We also use BGNP Theorem 5: if a localizable superoperator has eigenstates
\begin{equation}
|\psi\rangle,\quad (A\otimes\id)|\psi\rangle,\quad
(\id\otimes B)|\psi\rangle,
\label{eq:t5hyp}
\end{equation}
with $A$ and $B$ invertible, then $(A\otimes B)|\psi\rangle$ is also an eigenstate \cite{Beckman2001}.

For completeness, the following elementary observation identifies the pure eigenstates relevant here.

\begin{lemma}
For the rank-one decoherence map \eqref{eq:dephase}, a pure state $|\psi\rangle$ is an eigenstate in the BGNP sense, $\E(\ketbra{\psi})=\ketbra{\psi}$, if and only if $|\psi\rangle$ is proportional to one of the measurement-basis vectors $|\phi_x\rangle$.
\end{lemma}

\begin{proof}
Write $|\psi\rangle=\sum_x c_x|\phi_x\rangle$. Equation~\eqref{eq:dephase} removes every off-diagonal term $c_xc_y^*|\phi_x\rangle\langle\phi_y|$ with $x\neq y$. Equality with the original rank-one projector is therefore possible only when exactly one coefficient is nonzero. The converse is immediate.
\end{proof}

\section{Characterization theorem}

\begin{theorem}
A finite-dimensional bipartite rank-one ideal projective measurement is localizable in the sense of Definition 4 of Akibue and Miyazaki if and only if, up to LU equivalence, its measurement basis has the block-replicated nice-Bell form \eqref{eq:blockform}.
\end{theorem}

\begin{proof}
Sufficiency is the explicit localization protocol of Ref.~\cite{Akibue2026}. It remains to prove necessity.

\emph{1. A reference block is a nice Bell basis.---}
Fix $\B_{11}$ and choose $|\Phi_1\rangle\in\B_{11}$. Since all $d^2$ vectors in this block are maximally entangled on the same pair $\SA^1\otimes\SB^1$, choose Schmidt bases and unitaries $U_i$ on $\SA^1$ such that
\begin{equation}
|\Phi_i\rangle=(U_i\otimes\id)|\Phi_1\rangle,
\qquad U_1=\id,
\label{eq:Ui}
\end{equation}
for $i=1,\ldots,d^2$. Orthogonality gives
\begin{equation}
\mathrm{Tr}(U_i^\dagger U_j)=d\,\delta_{ij},
\label{eq:HS}
\end{equation}
so $\{U_i\}$ is a unitary error basis.

In the chosen Schmidt bases,
\begin{equation}
(U_j\otimes\id)|\Phi_1\rangle
=(\id\otimes U_j^T)|\Phi_1\rangle.
\label{eq:transpose}
\end{equation}
Extend $U_i$ and $U_j^T$ to the full local Hilbert spaces by direct sums with arbitrary unitaries on the orthogonal complements; in particular, the extensions can be chosen unitary and to preserve the reference subspaces. The three states
\begin{equation}
|\Phi_1\rangle,\quad
(U_i\otimes\id)|\Phi_1\rangle,\quad
(\id\otimes U_j^T)|\Phi_1\rangle
\end{equation}
are eigenstates of $\E$. BGNP Theorem 5 therefore implies that
\begin{equation}
(U_i\otimes U_j^T)|\Phi_1\rangle
=(U_iU_j\otimes\id)|\Phi_1\rangle
\label{eq:productstate}
\end{equation}
is also an eigenstate. Its support remains in $\SA^1\otimes\SB^1$, so the lemma and Eq.~\eqref{eq:blockbasis} imply
\begin{equation}
U_iU_j\propto U_k
\label{eq:closure}
\end{equation}
for some $k$ and every pair $i,j$. Indeed, equality of the corresponding vectors up to phase implies operator proportionality because $|\Phi_1\rangle$ has full Schmidt rank on the reference block. Since both sides are unitary, the proportionality factor has unit modulus. Thus $\{U_i\}$ satisfies the defining projective-closure condition for a nice unitary error basis \cite{Akibue2026}; hence $\B_{11}$ is a nice Bell basis.

\emph{2. Propagation along rows.---}
Fix $p$ and choose $|\Psi_p\rangle\in\B_{p1}$. Both $|\Psi_p\rangle$ and $|\Phi_1\rangle$ are maximally entangled and have the same Bob-side reduced state on $\SB^1$. Therefore there is a unitary isomorphism
\begin{equation}
W_p^A:\SA^1\rightarrow\SA^p
\end{equation}
such that
\begin{equation}
|\Psi_p\rangle=(W_p^A\otimes\id)|\Phi_1\rangle.
\label{eq:WpA}
\end{equation}
Because $\SA^1$ and $\SA^p$ have the same dimension, extend $W_p^A$ to a global unitary $\widetilde W_p^A$ on $\HA$. Only its restriction to $\SA^1$ enters the states below; the action on the orthogonal complement is arbitrary. For every $i$, Theorem 5 applied to
\begin{equation}
|\Phi_1\rangle,\quad
(\widetilde W_p^A\otimes\id)|\Phi_1\rangle,\quad
(\id\otimes U_i^T)|\Phi_1\rangle
\end{equation}
shows that
\begin{equation}
(\widetilde W_p^A\otimes U_i^T)|\Phi_1\rangle
=(W_p^A\otimes\id)|\Phi_i\rangle
\label{eq:rowstates}
\end{equation}
is an eigenstate in block $(p,1)$. These $d^2$ states are orthonormal and therefore exhaust that block:
\begin{equation}
\B_{p1}=\{(W_p^A\otimes\id)|\Phi_i\rangle\}_{i=1}^{d^2}.
\label{eq:row}
\end{equation}

\emph{3. Propagation along columns.---}
By the symmetric argument, for every $q$ there is a unitary isomorphism
\begin{equation}
W_q^B:\SB^1\rightarrow\SB^q
\end{equation}
with a global unitary extension such that
\begin{equation}
\B_{1q}=\{(\id\otimes W_q^B)|\Phi_i\rangle\}_{i=1}^{d^2}.
\label{eq:col}
\end{equation}

\emph{4. Propagation to every interior block.---}
For fixed $i,p,q$, let $\widetilde W_p^A$ and $\widetilde W_q^B$ be any global unitary extensions of the subspace isomorphisms. Eqs.~\eqref{eq:row} and \eqref{eq:col} show that
\begin{equation}
|\Phi_i\rangle,\quad
(\widetilde W_p^A\otimes\id)|\Phi_i\rangle,\quad
(\id\otimes \widetilde W_q^B)|\Phi_i\rangle
\end{equation}
are eigenstates of $\E$. A final application of Theorem 5 gives
\begin{equation}
(\widetilde W_p^A\otimes \widetilde W_q^B)|\Phi_i\rangle
=(W_p^A\otimes W_q^B)|\Phi_i\rangle
\label{eq:interiorstate}
\end{equation}
as an eigenstate in $\SA^p\otimes\SB^q$. Varying $i$ gives $d^2$ orthonormal states, which must exhaust the block:
\begin{equation}
\B_{pq}=\{(W_p^A\otimes W_q^B)|\Phi_i\rangle\}_{i=1}^{d^2}.
\label{eq:interior}
\end{equation}
Combining all blocks yields exactly Eq.~\eqref{eq:blockform}, up to the local choices of subspace bases. This proves necessity and completes the characterization.
\end{proof}

\section{Discussion}

The result closes the gap between the causal block structure of BGNP and the constructive family identified by Akibue and Miyazaki. Causality alone fixes the dimensions and maximal entanglement inside each block but permits, in principle, incompatible twists between blocks. Localizability is stronger: BGNP Theorem 5 forces products of independently admissible local transformations to produce another measurement eigenstate. Within a reference block this becomes projective closure of the unitary error basis; across blocks it forces the same nice Bell structure to propagate globally.

The proof also clarifies why no assumption that each block is independently localizable is needed. Every use of Theorem 5 concerns the \emph{full} decoherence superoperator $\E$. The block labels enter only through support: after the theorem produces a new eigenstate, its local supports identify the unique block in which it must lie.

The characterization is restricted to finite-dimensional bipartite rank-one ideal projective measurements. It does not address general POVMs, approximate localization, or multipartite measurements. Those extensions may require ingredients beyond the complete-measurement block structure used here.

\section*{Data availability}
No data were created or analyzed in this theoretical work.

\begin{acknowledgments}
OpenAI ChatGPT (GPT-5.6 Sol) was used substantively during exploratory scientific reasoning, literature organization, adversarial proof checking, and manuscript drafting. The author directed the use of the tool, checked the final argument against the cited primary sources, and takes full responsibility for the content of the manuscript.
\end{acknowledgments}

\end{document}